\documentclass[conference,letterpaper]{IEEEtran}

\usepackage[utf8]{inputenc} 
\usepackage[T1]{fontenc}
\usepackage{url}
\usepackage{ifthen}
\usepackage{cite}
\usepackage[cmex10]{amsmath} 
\usepackage{amsfonts,amssymb}
\usepackage{dsfont}

\usepackage{booktabs}
\usepackage{subcaption}
\usepackage{tikz}
\usepackage{pgfplots}
\usepackage{amsthm}

\newtheorem{claim}{Claim}
\newtheorem{definition}{Definition}
\pgfplotsset{compat=1.18}
\usepgfplotslibrary{groupplots}

\newcommand{\ffour}{\mathbb{F}_{2^2}}

\newcommand{\Sf}{\Sigma_4}
\newcommand{\D}{\mathsf{D}}
\newcommand{\Z}{\mathsf{Z}}
\newcommand{\p}{\text{Pr}}

\newcommand{\rvec}[1]{\boldsymbol{\mathsf{#1}}}
\newcommand{\x}[1]{\boldsymbol{#1}}
\newcommand{\xt}[1]{\tilde{\boldsymbol{#1}}}
\newcommand{\rv}[1]{\mathsf{#1}}

\usepackage[normalem]{ulem}

\allowdisplaybreaks
\IEEEoverridecommandlockouts

\begin{document}


\title{Marker-Delimited Codes for Short-Blocklength, High-Rate Coding over Multi-Read Edit Channels \thanks{This work was supported by the French government through the France 2030 investment plan managed by the National Research Agency (ANR) under reference ANR-22-CPJ2-0054-01.}} 
\author{\IEEEauthorblockN{\textbf{Sinan Ateş Yercan, Marc Antonini, Serge Kas Hanna}}
\IEEEauthorblockA{Côte d’Azur University, CNRS, I3S, Sophia Antipolis, France
\\
Emails: sinan.yercan@i3s.unice.fr, marc.antonini@cnrs.fr, serge.kas-hanna@cnrs.fr}
}

\maketitle

\begin{abstract}
   The read process of DNA-based data storage systems generates multiple noisy copies of the stored DNA sequences, affected by edit errors consisting of substitutions, deletions, and insertions. Motivated by the challenge of ensuring reliable data retrieval in the presence of edit errors, we present a concatenated coding scheme that accounts for practical design constraints in DNA storage. We introduce and apply the marker-delimited code~(MDC) as the inner code, which enables fast and reliable computation of symbolwise a posteriori probabilities~(APPs). We combine MDC with an outer LDPC code. The LDPC is decoded via belief propagation using the soft information generated by MDC. Our results show that, in comparison with prior work, this construction provides more efficient error correction over multi-read edit channels in the short-blocklength and high-rate~regime.    
\end{abstract}

\section{Introduction}
Synthetic deoxyribonucleic acid (DNA) has emerged as a promising medium for high-density and durable data storage, with several proof-of-concept demonstrations showing encouraging results over the last decade~\cite{church2012_setup, goldman2013towards, grass2015_setup, yazdi2017portable, chandak2019improved, press2020hedges, welzel2023dna, bar2025scalable, DNA-MGCP}. Despite this progress, data reliability remains a central challenge, largely due to the inherent errors and biases introduced by biochemical processes during synthesis, amplification, storage, and sequencing~\cite{heckel2019characterization,digitaltwin_errorrates}. Consequently, alongside ongoing improvements to the write (synthesis) and read (sequencing) processes for synthetic DNA molecules, commonly referred to as \emph{oligos}, the design and application of dedicated channel coding techniques have received increasing attention to enhance data reliability~\cite{milenkovic2024dna,sabary2024survey_coding}.

One of the main challenges that arises when using DNA as an alternative storage medium is its distinct error profile compared to conventional systems. In addition to substitution errors, which are common across many applications, DNA synthesis and sequencing introduce insertion and deletion errors~\cite{heckel2019characterization}. The combined presence of insertion, deletion, and substitution errors, collectively termed {\em edit} errors, has led to a growing body of work on adapting existing error-correcting codes (ECCs) to the DNA setting or designing new ones~\cite{cai2021correcting,tang2024correcting,hanna2024GC,mgcplus,press2020hedges,welzel2023dna}. Technical limitations in current synthesis technologies further complicate ECC design by imposing strict length constraints. Current synthesis methods are limited to short DNA oligos, typically a few hundred nucleotides in length, which remain costly to synthesize~\cite{heckel2019characterization}. 
As a result, writing data into DNA dominates the overall cost, motivating the need to keep storage overhead low and thus to minimize the redundancy introduced by ECCs. Taken together, these considerations make the design of {\em short-blocklength} and {\em high-rate} ECCs particularly important for reliable and cost-efficient DNA data storage.

Classical studies of edit-correcting codes predate the interest in DNA-based data storage and were conducted under different design objectives. A key line of work stems from modeling the encoder and channel jointly as a Hidden Markov Model (HMM), first proposed in~\cite{davey_IDS}, where the coding scheme helps restore the loss of synchronization induced by insertions and deletions. Building on this foundation, convolutional codes~\cite{mansour2010convolutional} and marker-based constructions~\cite{ratzer2005marker} have received significant attention as synchronization codes. In practice, such synchronization codes are typically used as inner codes and combined with a reliable substitution-correcting outer code, such as an LDPC code, within a {\em concatenated} coding scheme.

In the context of DNA storage in particular, an additional and fundamental consideration is decoding from multiple reads. The read process (sequencing) in DNA storage produces several noisy copies of the synthesized oligos, and leveraging this multi-read structure can substantially improve decoding performance compared to using a single read. This problem is often referred to as the trace reconstruction problem, first introduced in~\cite{batu2004reconstructing}, and has been the focus of extensive theoretical study~\cite{cheraghchi2020_codedtrec,chase2021approximate,davies2021approximate}. From a coding-theoretic perspective, several works have explored practical multi-read decoding by extending the HMM framework to jointly process multiple traces~\cite{conv_concat,symbolwiseMAP,banerjee2024sequential,banerjee2025sequential}. Other interdisciplinary approaches combine channel coding with multiple sequence alignment (MSA) techniques~\cite{TBMA}. While MSA algorithms are also commonly used on their own in practice, their effectiveness is limited to correcting sequencing errors, whereas coding-based methods can additionally correct synthesis errors that appear systematically across reads~\cite{press2020hedges}. Alternatively, multiple reads can first be aligned to form a consensus sequence, after which an edit-correcting code, such as the recently introduced MGC+ code~\cite{DNA-MGCP}, can be used to correct the residual edit errors. However, multi-read coding schemes can jointly exploit the information contained across all reads for improved decoding performance.

In this work, we present a concatenated coding scheme based on an inner code that we introduce, called the marker-delimited code (MDC). We describe the MDC construction and detail its decoding algorithm. The MDC decoder exploits marker sequences to divide the decoding problem over multiple full reads into smaller independent subproblems over shorter read fragments. These subproblems are solved separately and in parallel, leading to faster decoding overall. Moreover, since each subproblem involves only short fragments, the decomposition makes likelihood-based decoding computationally tractable. We realize this idea by introducing a soft-output decoder that operates on a segment-by-segment basis and uses the likelihoods to compute symbolwise a posteriori probabilities (APPs). The inner MDC code is paired with an outer LDPC code, which leverages these APPs through belief-propagation decoding. In comparison with existing approaches, we show that our coding scheme achieves both faster and more reliable decoding in the short-blocklength, high-rate regime of interest.


\section{Preliminaries}
\subsection{Notation}
We denote vectors with bold, random variables with sans-serif, and sets with calligraphic letters, e.g. $\x{x}$, $\rv{X}$, and $\rvec{X}$ denote a vector, random variable, and a random vector, respectively. For integers $i,j$, and $n$, we define \mbox{$[n] \triangleq \{1,2,\dots,n\}$}, and \mbox{$[i,j] \triangleq \{i,i+1,\dots,j\}$} for $i\leq j$ and $[i,j] \triangleq \{\}$ otherwise. For a vector $\x{x}$, the notation $\x{x}_{[i,j]}$ denotes the subvector $(x_i,x_{i+1},\dots,x_j)$. We use $\x{x}^{[t]} \triangleq (\x{x}^{1},\dots,\x{x}^{t})$ to denote a collection of vectors, where each $\x{x}^{i}$ may have arbitrary length. 
The notation $\x{x}||\x{y}$ denotes vector concatenation. We use $\mathbb{F}_q$ to denote the field of size~$q$, and $\Sigma_q = \{0,1,\dots,q-1\}$ for the alphabet of size $q$. 

\subsection{Multi-Read Edit Channel}\label{sec:ch_mod}
We consider a channel that takes as input a codeword $\boldsymbol{x}\in\Sigma_4^n$ and outputs $t$ reads (noisy copies) of it, denoted by $\x{y}^{[t]} = (\x{y}^1,\ldots,\x{y}^t)$. Each read $\x{y}^i\in\Sigma_4^*$ is generated independently by passing $\boldsymbol{x}$ through a random edit channel that operates sequentially on the input symbols. At each position, the channel applies one of four events with probabilities $p_i, p_d, p_s$, and $p_r$: insertion, deletion, substitution, or retention. For an input symbol $x\in \Sigma_4$, the observed channel output is $\x{y} = (\sigma,x)$ for an insertion event with $\sigma$ chosen uniformly at random from~$\Sf$, $y = \bar{x}$ for a substitution event with $\bar{x}$ chosen uniformly at random from~$\Sf\setminus \{x\}$, and $y=x$ for the retention event. In case of a deletion, the input symbol is removed and does not appear in the channel output. In this channel model, consecutive insertions are limited to length one, and inserted symbols are not edited again. This corresponds to a simplified version of the IDS channel~\cite{davey_IDS}. This simplification aligns with experimentally validated error statistics in DNA storage, which show that insertions are rare under standard pipelines~\cite{digitaltwin_errorrates}.

We use the concept of \textit{drift}~\cite{davey_IDS} to quantify the impact of the insertion-deletion errors. For channel input and output sequences $\boldsymbol{x},\boldsymbol{y}$, the drift vector $\x{d} \in \mathbb{Z}^{n+1}$ is the measure of the synchronization loss between $\boldsymbol{x}$ and $\boldsymbol{y}$. Specifically, we take $d_i$ as the cumulative difference between the total number of insertion and deletion events after the transmission of the \mbox{$i$-th} input symbol~$x_i$. The drift vector $\boldsymbol{d}=(d_0,d_1,\ldots,d_n)$ thus forms a Markov chain. For $i = 0,1,\ldots,n-1$, the transition probabilities are given by
\begin{equation}\label{eqn:drift_transition_prob}
    \p(d_{i+1}|d_i) =  \begin{cases}
        p_i &,\; d_{i+1} = d_i + 1 \\[-3pt]
        p_r + p_s &,\; d_{i+1} = d_i \\[-3pt]
        p_d &,\; d_{i+1} = d_i - 1  
    \end{cases},
\end{equation} where $d_0 = 0$ with probability one and $\lvert d_{i+1} - d_i\rvert \leq 1$. 
\section{Coding Scheme}
We propose a concatenated coding scheme composed of an inner code used for generating soft information, and an outer error-correcting code that utilizes this soft information to decode the message. As an inner code, we introduce the Marker-Delimited Code (MDC), which produces symbolwise APPs over $\Sigma_4$ under the multi-read edit channel. As shown later in Section~\ref{simul}, at high rates, MDC offers both faster decoding and improved reliability relative to commonly used inner codes for multi-read channels, such as convolutional~\cite{conv_concat,banerjee2024sequential,banerjee2025sequential} and marker-repeat (MR) codes~\cite{TBMA}. For the outer code, we consider LDPC codes and adapt their soft decoding procedures to the quaternary setting to match the APPs produced by MDC. 
\subsection{Marker-Delimited Codes (MDC)}
We now introduce the inner MDC code. Informally, an MDC code partitions the information message into $\nu$ blocks, encodes each block
independently using a short code $\mathcal{C}$, and appends a fixed marker sequence to each encoded block. The markers provide a known structure within the codeword that the decoder exploits to generate symbolwise APPs for the input message. The formal definition is given below.
\begin{definition}\label{def}
An \mbox{$[\tilde{n},\tilde{k},\mu,\nu]$-MDC} code is defined by a marker sequence $\x{m}\in\Sigma_4^\mu$, and a code
$\mathcal{C}\subseteq\Sigma_4^{\tilde{n}}$ of size $4^{\tilde{k}}$ with encoding function \mbox{$\mathcal{E}:\Sigma_4^{\tilde{k}}\rightarrow\Sigma_4^{\tilde{n}}$}. The MDC encoder maps an information
message $\x{u}\in\Sigma_4^k$, with $k=\nu\tilde{k}$, into a codeword
$\x{x}\in\Sigma_4^n$ of the form
\[
\x{x} =
\big(
\mathcal{E}(\x{u}_{[1,\tilde{k}]}) \| \x{m} \|\; \ldots \; \|
\mathcal{E}(\x{u}_{[(\nu-1)\tilde{k}+1,k]}) \| \x{m}
\big),
\]
where $B\triangleq \tilde{n}+\mu$ is the \emph{segment length}, $n=\nu B$, and $R = \tilde{k}/B$ is the code rate. The MDC decoder takes as input $t$ reads of $\x{x}$ over the multi-read edit channel and outputs symbolwise APPs for $\x{u}\in\Sigma_4^k$.
\end{definition}

The purpose of the markers is to {\em delimit} the codeword into non-overlapping segments that can be identified and isolated at the decoder. This segmentation allows the full-length multi-read APP computation problem to be divided into $\nu$ smaller subproblems over shorter segments. These subproblems can then be processed independently and in parallel, resulting in faster decoding. Furthermore, the MDC parameters are chosen such that, within each subproblem, the segment length $B$ is sufficiently small so that likelihood-based APP computation becomes tractable. The MDC decoding procedure that exploits this structure is a central component of the proposed construction and is described in detail in Section~\ref{sec:decoding}.

\subsection{Concatenated Architecture}
The concatenated scheme uses an LDPC code as the outer code and MDC as the inner code. During encoding, the information message is first encoded by the LDPC encoder, and the resulting LDPC codeword is then passed to the MDC encoder. The parity-check matrix of the LDPC code is generated using the progressive edge-growth (PEG) method for constructing regular LDPC codes~\cite{peg_theo,peg_implement}. Random nonzero elements from $\ffour$ are then assigned to the nonzero entries of the parity-check matrix, yielding a nonbinary LDPC over $\ffour$.

At the decoder, the MDC decoder processes $t$ channel reads and computes symbolwise APPs for the quaternary symbols of the LDPC codeword. These APPs are then used as soft input to a belief-propagation (BP) decoder, adapted to accommodate arithmetic over the field $\ffour$~\cite{Declercq2007NBLDPC}. 

\section{MDC Soft-Output Decoding}\label{sec:decoding} 
We begin by providing a high-level overview of the MDC decoding process, while detailed technical descriptions of its constituent steps are given in the following subsections. The MDC decoder follows a divide-and-conquer type approach composed of three sequential steps: \emph{segmentation}, \emph{likelihood-based APP computation}, and \emph{confidence-based boundary refinement}. The decoder first uses the static markers to split each of the $t$~reads into individual segments. It then computes symbolwise APPs for each segment by evaluating the likelihoods of a shortlisted set of candidate codewords using the corresponding read fragments. Finally, the decoder refines the segmentation boundaries based on the per-segment confidences obtained from the computed likelihoods and updates the APPs accordingly.

 
The key advantage of this decoding approach is that the segmentation step reduces the original multi-read decoding problem into smaller subproblems over short read fragments. This makes likelihood-based APP computation tractable while still allowing the read fragments associated with each segment to be processed jointly. In contrast, graph-based algorithms such as BCJR require an exponential complexity in the number of reads $t$ when used for exact joint multi-read inference~\cite{conv_concat}.
To avoid this exponential complexity, running the BCJR algorithm separately on each read and subsequently multiplying the resulting APPs has been proposed~\cite{conv_concat, 10161631}, while other state-of-the-art methods, such as Trellis BMA~\cite{TBMA}, approximate joint multi-read BCJR inference using interacting single-read trellises. Both approaches avoid exact joint inference across all reads, which results in degraded APPs.
An additional advantage of the segmentation step is that the resulting subproblems can be processed independently and in parallel, thereby reducing the overall decoding time compared to BCJR-based methods.

\subsection{Segmentation}
Let $\rvec{D} = (\D_0,\D_1,\dots,\D_{n})$ be the random drift vector (defined in Section~\ref{sec:ch_mod}) associated with the transmission of an MDC codeword \mbox{$\x{x} \in \Sigma_4^{n}$} over the edit channel, resulting in a read sequence $\x{y} \in \Sigma_4^{n'}$. Let $\rvec{Z} = (\Z_0,\Z_1,\dots,\Z_{\nu})$ be the offset vector, where $\Z_j = \D_{jB}$, $j\in [\nu]$, corresponds to the net drift at the end of each MDC segment of length~$B$, with $\Z_0=0$. Then, the segmentation task is equivalent to estimating~$\rvec{Z}$. The difference $\rv{Z}_j - \rv{Z}_{j-1}$ is exactly the net offset introduced within segment~$j$. Let $d = z_j - z_{j-1}$, then
\begin{equation}\label{eqn:offset_recusion}
    \p(z_j|z_{j-1}) = \sum_{\delta = -\mu}^{\mu}\p(\Z^{\text{cw}}_{j} = d - \delta)\p(\Z^{\text{mar}}_{j} = \delta),
\end{equation}
describes the probability of the segment offset $z_j$, given the previous segment's offset $z_{j-1}$, where $\Z^\text{cw}_{j}$ is the net offset in the $\tilde{n}$ symbols of the codeword from $\mathcal{C}$, and $\Z^\text{mar}_{j}$ is the net offset in the $\mu$ marker symbols. The probability mass function (PMF) of the offset in the codeword symbols, $\p(\Z^\text{cw}_{j} = d)$, can be expressed as
\begin{align} 
    \p(\Z^\text{cw}_{j} = d) = \sum_{\eta = \max\{0,-d\}}^{\lfloor \frac{\tilde{n}-d}{2}\rfloor}&{\binom{\tilde{n}}{\eta ,d+\eta ,\tilde{n}-2\eta -d}}\nonumber\\[-5pt]
    &\quad \times {p_d}^{\eta }{p_i}^{d+\eta}{(p_r+p_s)}^{\tilde{n}-2\eta-d},
\end{align}
by counting the number of error patterns of length $\tilde{n}$ that yield the desired net offset $d$. To obtain the distribution of the offset in the marker, $\p(\Z^\text{mar}_{j} = \delta)$, we make use of the structure of the marker $\x{m}$, and segment offset $z_j$. Let $$\tilde{\x{m}} \triangleq \x{y}_{[jB+z_j-\delta-\mu+1,\; jB+z_j-\delta]}$$ be the {\em received marker} of length~$\mu$. Then, we have \mbox{$\p(\Z^\text{mar}_{j} = \delta) = \p(\tilde{\x{m}}| \x{m})$}, which can be formulated analytically as in~\cite{markertrellis_decoding,mgcplus}.

For short $\x{m}$ ($\mu \leq 4$), $\p(\tilde{\x{m}}| \x{m})$ can alternatively be precomputed for all possible pairs of $(\tilde{\x{m}},\delta)$ by iterating over all possible error patterns $\x{e}$ that can affect the marker $\x{m}$. Let~$\mathcal{M}_i$ be the multiset of all possible sequences obtained from the first $i$ symbols of $\x{m}$, and $\mathcal{W}(m_{i+1},e_i)$ be the set of all possible channel outputs of the symbol $m_{i+1}$ due to error $e_i$. Then $\mathcal{M}_{i+1} = \{\x{r}||\x{s}:\x{r}\in\mathcal{M}_i,\x{s}\in\mathcal{W}(m_{i+1},e_i)\}$, where we conclude that any sequence $\tilde{\x{m}} \in \mathcal{M}_i$, with multiplicity $m(\tilde{\x{m}})$, has probability $\p(\tilde{\x{m}}|\x{m},\x{e}) = \p(\x{e})$. The error pattern $\x{e}$ directly yields the net marker offset $\delta$, and if the received $\tilde{\x{m}}$ is shorter than $\mu$, we modify the received markers as $\tilde{\x{m}}' = (\tilde{\x{m}}||\x{m}''), \forall \x{m}'' \in \Sf^{\mu - |\tilde{\x{m}}|}$ with probability $\p(\x{e})\times4^{|\tilde{\x{m}}| - \mu}$. If $|\tilde{\x{m}}| \geq \mu$, then we take $\tilde{\x{m}}' = \tilde{\x{m}}_{[\mu]}$, with associated probability $\p(\x{e})$.

Using~\eqref{eqn:offset_recusion}, the offsets $\Z_i$ can be calculated recursively through a Viterbi-like algorithm, similar to~\cite{mgcplus}. We begin by constructing a trellis with $\nu+1$ stages, where stage $j \in [0,\nu]$ corresponds to $\Z_j$, and each stage contains nodes associated with the allowable drift values $z_j \in [-\Phi,\Phi]$. The node at stage $j$ and value $z_j$ is assigned a path-metric: 
\begin{equation}\label{eqn:seg_path_metric}
    \alpha_j(z_j) = \max_{z_{j-1}} \alpha_{j-1}(z_{j-1})\p(z_j|z_{j-1}),
\end{equation}
and a predecessor 
\begin{equation}
    \text{pred}_j(z_j) = \arg\max_{z_{j-1}}\alpha_{j-1}(z_{j-1})\p(z_j|z_{j-1}).
\end{equation}
Both the metrics and the predecessors are calculated recursively from layer 1 to $\nu$, with $\alpha_0(z_0) = \mathds{1}_{\{z_0=0\}}$ and \mbox{$\alpha_\nu(z_\nu) = \mathds{1}_{\{z_\nu = n'-n\}}$}, since $\Z_0 = \D_0 = 0$ and \mbox{$\Z_{\nu} = \D_n = n'-n$}. After each layer is fully calculated, the path-metrics are normalized to ensure numerical stability. Then, the estimates are obtained as \mbox{$\rv{\hat{Z}}_{j} = \text{pred}_{j+1}(z_{j+1})$}.

\subsection{Shortlisted Likelihood-Based APP Computation }\label{sec:shortlistedMAP}

Recall that the MDC construction (Definition~\ref{def}) is defined by a code $\mathcal{C}\subseteq \Sigma_4^{\tilde{n}}$ with encoding function $\mathcal{E}:\Sigma_4^{\tilde{k}}\to\Sigma_4^{\tilde{n}}$, and a marker sequence $\x{m}\in\Sigma_4^\mu$. We now define the extended code $\mathcal{C}'$ obtained by appending the marker sequence $\x{m}$ to each codeword of~$\mathcal{C}$, i.e.,
\[
\mathcal{C}' \triangleq \{ \tilde{\x{x}} = \mathcal{E}(\tilde{\x{u}})\|\x{m} : \tilde{\x{u}}\in\Sigma_4^{\tilde{k}} \} \subseteq \Sigma_4^B,
\]
where \mbox{$\tilde{\x{u}}=(\tilde{u}_1,\dots,\tilde{u}_{\tilde{k}})$} denotes the message substring associated with an MDC segment of length $B=\tilde{n}+\mu$. Let \mbox{$\tilde{\x{y}}^{[t]}=(\tilde{\x{y}}^1,\tilde{\x{y}}^2,\dots,\tilde{\x{y}}^t)$} denote the read fragments corresponding to a given segment. The posterior probability of $\tilde{\x{x}}\in\mathcal{C}'$ given the read fragments is
\begin{align} \label{eq:post}
    \p(\tilde{\x{x}}|\tilde{\x{y}}^{[t]})
    =
    \frac{\p(\tilde{\x{y}}^{[t]}|\tilde{\x{x}})\p(\tilde{\x{x}})}
    {\sum_{\tilde{\x{x}}'\in\mathcal{C}'} \p(\tilde{\x{y}}^{[t]}|\tilde{\x{x}}')\p(\tilde{\x{x}}')}.
\end{align}
Assuming a uniform prior over $\mathcal{C}'$, this reduces to
\begin{equation}
    \p(\tilde{\x{x}}|\tilde{\x{y}}^{[t]})
    =
    \frac{\p(\tilde{\x{y}}^{[t]}|\tilde{\x{x}})}
    {\sum_{\tilde{\x{x}}'\in\mathcal{C}'} \p(\tilde{\x{y}}^{[t]}|\tilde{\x{x}}')}.
\end{equation}
Furthermore, since the $t$ reads are obtained independently for a given channel input $\xt{x}$, we have
\begin{equation}
    \p(\tilde{\x{y}}^{[t]}|\tilde{\x{x}})
    =
    \prod_{\tau=1}^t \p(\tilde{\x{y}}^\tau|\tilde{\x{x}}),
\end{equation}
where each per-read likelihood $\p(\tilde{\x{y}}^\tau|\tilde{\x{x}})$ can be computed in $\mathcal{O}(B^2)$ time using dynamic programming (DP) as in~\cite{markertrellis_decoding}.

If the segment boundaries have been correctly estimated in the previous step and the posterior in \eqref{eq:post} is computed for all $\tilde{\x{x}}\in\mathcal{C}'$, then the \emph{exact} symbolwise APPs of $\tilde{\x{u}}\in\Sigma_4^{\tilde{k}}$ can be obtained via marginalization. Namely, for $i\in[\tilde{k}]$ and~\mbox{$a\in\Sigma_4$}, the symbolwise APPs are given by
\begin{align} \label{eqn:likelihoodcodebook}
    \p(\tilde{u}_i=a|\tilde{\x{y}}^{[t]})
    =
   \sum_{\tilde{\x{x}}\in\mathcal{C}',\,\tilde{u}_i=a}
\p(\tilde{\x{x}}|\tilde{\x{y}}^{[t]}).
\end{align}

\enlargethispage{-5pt}

The exact evaluation of \eqref{eqn:likelihoodcodebook} therefore necessitates computing $\p(\tilde{\x{x}}|\tilde{\x{y}}^{[t]})$ for all $|\mathcal{C}'|=4^{\tilde{k}}$ codewords, each involving a quadratic-time DP computation in the segment length~$B$. To reduce this complexity, the MDC decoder first constructs a shortlist of candidates $\mathcal{L}\subseteq \mathcal{C}'$ using a similarity-based metric computable in linear time. Specifically, for each read fragment $\tilde{\x{y}}^\tau$, a per-read shortlist $\mathcal{L}_\tau\subseteq\mathcal{C}'$ is formed by selecting the candidates $\tilde{\x{x}}$ with the largest $\kappa$-mer-based Jaccard similarity to $\tilde{\x{y}}^\tau$, defined as
\begin{equation}
    J(\tilde{\x{x}},\tilde{\x{y}}^\tau)
    \triangleq
    \frac{|\mathcal{K}(\tilde{\x{x}})\cap\mathcal{K}(\tilde{\x{y}}^\tau)|}
    {|\mathcal{K}(\tilde{\x{x}})\cup\mathcal{K}(\tilde{\x{y}}^\tau)|},
\end{equation}
where $\mathcal{K}(\tilde{\x{x}})=\{\tilde{\x{x}}_{[i,i+\kappa-1]}: i\in[B-\kappa+1]\}$ denotes the set of $\kappa$-mers of $\tilde{\x{x}}$. The final shortlist is obtained as \mbox{$\mathcal{L}=\cup_{\tau=1}^t \mathcal{L}_\tau$}. The APPs are then \emph{approximated} by restricting the computations in~\eqref{eq:post} and \eqref{eqn:likelihoodcodebook} to $\mathcal{L}$ instead of $\mathcal{C}'$. 

Therefore, instead of applying the quadratic-time DP likelihood computation to all $|\mathcal{C}'| = 4^{\tilde{k}}$ candidates, the decoder first performs a linear-time Jaccard similarity evaluation over~$\mathcal{C}'$ to construct a shortlist $\mathcal{L}$, and then applies the quadratic-time posterior computation only to the $L\triangleq |\mathcal{L}|$ shortlisted candidates. Since Jaccard similarity over $\kappa$-mers approximates sequence similarity under edit errors~\cite{ukkonen1992approximate}, the shortlist favors candidates with small edit distances and large posterior probabilities. The computed APPs therefore remain accurate provided that the shortlist captures most of the posterior mass, i.e., $\sum_{\tilde{\x{x}}\in\mathcal{L}} \p(\tilde{\x{x}}|\tilde{\x{y}}^{[t]})$ is close to~$1$.

\subsection{Confidence-Based Boundary Refinement}
The accuracy of the APPs obtained in the previous step also depends strongly on the accuracy of the segmentation step. Therefore, after the initial APP computation, the decoder refines the segmentation boundaries. For segment $j\in [\nu]$, we define the confidence score $F_j$ as
\begin{equation}
 F_j \triangleq \frac{\max_{\xt{x}: \xt{x} \in \mathcal{L}}\p(\xt{y}^{[t]}|\xt{x})}{\sum_{\xt{x} \in \mathcal{L}}\p(\xt{y}^{[t]}|\xt{x})},
\end{equation}
with $\xt{x}^*$ being the maximizing argument. Then, for segments satisfying $F_j \geq F_T$, where $F_T \in [0,1]$ is a confidence threshold, the corresponding probabilities in \eqref{eqn:offset_recusion} are updated to incorporate the APP computation from the previous step as follows:
\begin{multline}
    \p^{\text{u}}(z_j|z_{j-1}) =
    \sum_{\delta = -\mu}^{\mu}\p(\xt{y}_{[(j-1)B + z_{j-1},jB + z_j - \delta]}\mid \xt{x}^*)\\[-0.75em] \times\p(\Z^{\text{mar}}_{j} = \delta).
\end{multline}
The segmentation and APP computation steps are then repeated using the updated transition probabilities, while retaining the original shortlisted candidates. Finally, if the confidence of a segment increases after refinement, the decoder updates the corresponding APPs; otherwise, it retains the APPs obtained from the initial segmentation.

\begin{figure*}[t]
    \hspace{-10pt}
    \subfloat[SER, $R=0.7$]{\label{fig:1a}
        \begin{tikzpicture}
        \begin{axis}[
                width=0.34\textwidth,  
                height=0.3\linewidth, 
                ymode=log,
                grid=both,
                major grid style={gray!50},
                minor grid style={gray!30},
                ylabel style = {yshift = -4pt},
                ytick={1e-5,1e-4,1e-3,1e-2,1e-1,1},
                xtick = {1,2,3,4,5,6},
                ymin=1e-4,
                ymax=4e-1,
                xmin=1,
                xmax=6,
                tick label style={font=\fontsize{7}{8}\selectfont},
                label style = {font = \fontsize{9}{10}\selectfont},
                xlabel = {Number of Reads $t$},
                ylabel = {Symbol Error Rate (SER)},
                legend style = {font = \fontsize{6}{7}\selectfont,fill opacity = 0.8,
                        draw opacity = 0.8, text opacity = 1, anchor = north east, at = {(0.99,0.99)}
                        ,cells = {anchor = west}, nodes = {inner sep = 1pt}}
        ]
        \addplot[red, thick, mark=diamond, mark options={fill=white}] coordinates {
            (1,0.1556)(2,0.0501)(3,0.01626)(4,0.0051)(5,0.00186)(6,0.00066)
        };
       \addlegendentry{CC-TBMA}
        
        \addplot[teal, thick,mark=diamond,mark options = {fill=white}] coordinates{
            (1,0.11791071428571429)(2,0.04903571428571428)(3,0.016339285714285712)(4,0.006366071428571427)
                        (5,0.0026875)(6,0.001044642857142857)
        };
        \addlegendentry{MR-TBMA}

        
        \addplot[blue, thick,mark = diamond, mark options = {fill = white}] coordinates{
            (1,0.165)(2,0.0581)(3,0.012)(4,0.00225)(5,0.000555)
                            (6,0.000153)
        };
        \addlegendentry{MDC (this work)}
        
        \end{axis}
        \end{tikzpicture}
    
    }
    \subfloat[SER, $R=0.8$]{\label{fig:1b}
        \begin{tikzpicture}
        \begin{axis}[
                width=0.34\textwidth,  
                height=0.3\linewidth, 
                ymode=log,
                grid=both,
                major grid style={gray!50},
                minor grid style={gray!30},
                ytick={1e-5,1e-4,1e-3,1e-2,1e-1,1},
                xtick = {1,2,3,4,5,6,7,8},
                ymin=1e-4,
                ymax=4e-1,
                xmin=1,
                ylabel style = {yshift = -4pt},
                xmax=8,
                tick label style={font=\fontsize{7}{8}\selectfont},
                label style = {font = \fontsize{9}{10}\selectfont},
                 ylabel = {Symbol Error Rate (SER)},
                xlabel = {Number of Reads $t$},
                legend style = {font = \fontsize{6}{7}\selectfont,fill opacity = 0.8,
                        draw opacity = 0.8, text opacity = 1, anchor = north east, at = {(0.99,0.99)}
                        ,cells = {anchor = west}, nodes = {inner sep = 1pt}}
        ]
        \addplot[red,thick,mark=o,mark options = {fill=white}]coordinates{
            (1,0.229578125)(2,0.1028984375)(3,0.04884375)(4,0.0194453125)(5,0.0078671875)(6,0.003796875)(7,0.00165625)(8,0.0008515625)
        };
        \addlegendentry{CC-TBMA}
                
        \addplot[teal,thick,mark=o,mark options = {fill=white}]coordinates{
            (1,0.150515625)(2,0.0748671875)(3,0.030203125)(4,0.0138359375)(5,0.006875)(6,0.003640625)(7,0.0020234375)(8,0.0012421875)        };
        \addlegendentry{MR-TBMA}

        \addplot[blue,thick,mark=o,mark options = {fill = white}]coordinates{
            (1,0.161597)(2, 0.091888)(3,0.026882)(4,0.00879140625)(5,0.00276796875)(6,0.0008625)(7,0.00026875)(8,0.0001125)
        };
        \addlegendentry{MDC (this work)}
        \end{axis}
        \end{tikzpicture}
    }
    \subfloat[Decoding Time, $R=0.8$]{\label{fig:1c}
\begin{tikzpicture}
\begin{axis}[
width=0.34\textwidth,  
                height=0.3\linewidth,
    xlabel={Number of Reads $t$},
    ylabel={Average Decoding Time (s)},
    xmin=1, xmax=8,
    ylabel style = {yshift = -4pt},
    ymin=0, ymax=1.4,
    xtick={1,2,3,4,5,6,7,8},
    ytick={0, 0.2, 0.4, 0.6, 0.8, 1.0, 1.2, 1.4},
    legend style = {font = \fontsize{6}{7}\selectfont,fill opacity = 0.8,
                        draw opacity = 0.8, text opacity = 1, anchor = north west, at = {(0.01,0.99)}
                        ,cells = {anchor = west}, nodes = {inner sep = 1pt}},
    grid=both,
                    major grid style={gray!50},
                minor grid style={gray!30},
                    tick label style={font=\fontsize{7}{8}\selectfont},
                label style = {font = \fontsize{9}{10}\selectfont},
]


\addplot[
    color=teal,
    thick,
    mark=o,
    mark options = {fill=white}
] coordinates {
    (1, 0.0772)
    (2, 0.1489)
    (3, 0.2200)
    (4, 0.2937)
    (5, 0.3675)
    (6, 0.4341)
    (7, 0.5030)
    (8, 0.5736)
};
\addlegendentry{MR-TBMA}


\addplot[
    color=red,
    thick,
    mark=o,
    mark options = {fill=white}
] coordinates {
    (1, 0.1627)
    (2, 0.3188)
    (3, 0.4757)
    (4, 0.6357)
    (5, 0.7870)
    (6, 0.9369)
    (7, 1.0982)
    (8, 1.2558)
};
\addlegendentry{CC-TBMA}

\addplot[
    color=blue,
    thick,
    mark=o,
    mark options = {fill=white, solid}
] coordinates {
    (1, 0.069801)
    (2, 0.107822)
    (3, 0.145873)
    (4, 0.183724)
    (5, 0.218002)
    (6, 0.243673)
    (7, 0.272989)
    (8, 0.301161)
};
\addlegendentry{MDC (serial)}

\addplot[
    color=blue,
    thick,
    mark=o,
    mark options = {fill=white, solid},
    dashed
] coordinates {
    (1, 0.017972)
    (2, 0.021038)
    (3, 0.024144)
    (4, 0.028026)
    (5, 0.030904)
    (6, 0.032048)
    (7, 0.032439)
    (8, 0.034046)
};
\addlegendentry{MDC (parallel)}
\end{axis}
\end{tikzpicture}
}
        \caption{Symbol error rate after applying hard-decision decoding to the APPs obtained from different quaternary inner codes of rate $R\in \{0.7, 0.8\}$ and common blocklength $n=160$, along with the corresponding average decoding times. The channel edit error rate is fixed to $p_e = 10\%$. The parallel decoding times reported in (c) correspond to execution on 16 CPU cores. All results are averaged over $10^5$ independent trials.}\label{fig:compare_inner}
\end{figure*}
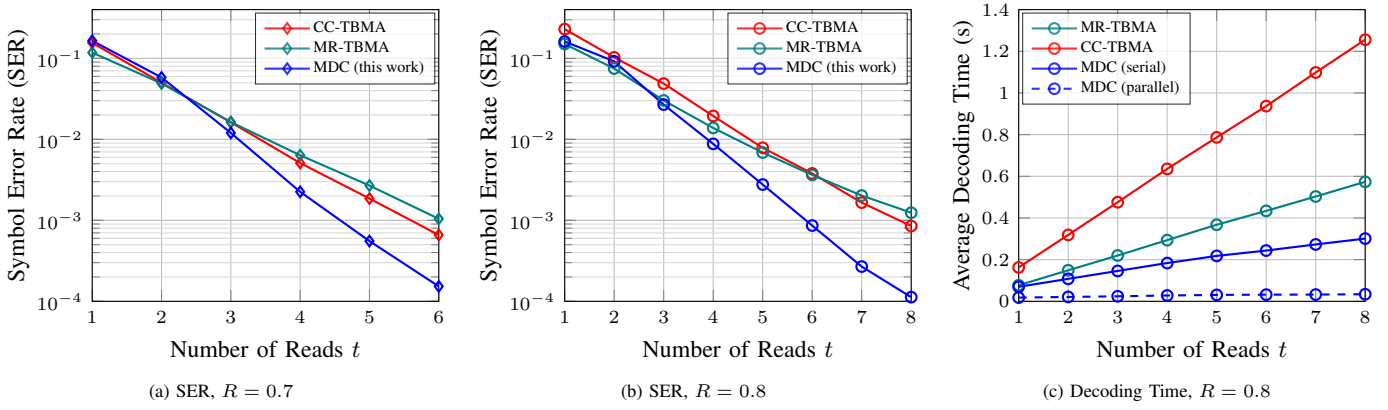

\subsection{Computational Complexity}
Next, we analyze the complexity of the sequential decoding steps. The overall complexity of the MDC decoding algorithm is stated in Claim~\ref{cl:SLL_comp} and follows from summing the complexities of its constituent components.

\begin{claim}\label{cl:SLL_comp}
For an \mbox{$[\tilde{n},\tilde{k},\mu,\nu]$-MDC} code decoded from $t$ reads, the total computational complexity is \mbox{$\mathcal{O}(t\nu\Phi^2 + t\nu B4^{\tilde{k}} + t\nu B^2L)$}, where $\Phi$ is the offset limit, $B$ is the segment length, and $L$ is the shortlist size. With parallelization across reads in the segmentation step and across segments in the APP computation step, the parallel time complexity is $\mathcal{O}(\nu\Phi^2 + tB4^{\tilde{k}} + tB^2L)$.
\end{claim}

\begin{proof}
For the segmentation step, the main computational cost comes from evaluating \eqref{eqn:offset_recusion} for each pair $(z_{j-1},z_j)$ and each $j \in [\nu]$. Since the algorithm restricts each offset to the range $z_j \in [-\Phi,\Phi]$, there are at most $(2\Phi)^2$ offset pairs to consider at each segment boundary. Therefore, the complexity of segmentation is $\mathcal{O}(\nu\Phi^2)$ for a single read and $\mathcal{O}(t\nu\Phi^2)$ for $t$ reads. Since the reads can be processed independently, this step is parallelizable across reads, reducing the effective parallel time complexity to $\mathcal{O}(\nu\Phi^2)$.

The APP computation step consists of two parts: shortlisting based on Jaccard similarity over the full codebook~$\mathcal{C}'$, and posterior-probability calculation over the shortlisted candidates $\mathcal{L}\subseteq \mathcal{C}'$. Computing the Jaccard similarity between two vectors of length $B$ requires identifying the corresponding $\kappa$-mers and computing the intersection and union of the resulting $\kappa$-mer sets, all of which can be done in $\mathcal{O}(B)$ time for fixed $\kappa$. Since this computation is performed for each of the $|\mathcal{C}'|=4^{\tilde{k}}$ codewords, $\nu$ segments, and $t$ reads, the shortlisting complexity is $\mathcal{O}(t\nu B 4^{\tilde{k}})$. Because segments can be processed independently, the effective parallel time complexity of 
shortlisting reduces to $\mathcal{O}(tB4^{\tilde{k}})$. The DP-based probability computation incurs a complexity of $\mathcal{O}(B^2)$~\cite{markertrellis_decoding} per (read, segment, shortlisted codeword) tuple. Hence, for $t$ reads, $\nu$ segments, and a shortlist of size $|\mathcal{L}|=L$, the total computational complexity of the posterior probability calculation is $\mathcal{O}(t\nu B^2 L)$. Similar to shortlisting, the corresponding parallel time complexity reduces to $\mathcal{O}(tB^2L)$.

The confidence-based boundary refinement executes essentially the same segmentation and APP computation steps, but without shortlisting. Its total computational complexity is therefore $\mathcal{O}(t\nu\Phi^2 + t\nu B^2L)$, and its parallel time complexity is $\mathcal{O}(\nu\Phi^2 + tB^2L)$. These terms are already present in the complexity of the initial segmentation and posterior-probability calculation, and hence the refinement step does not change the overall complexity order.
\end{proof}
While linear in the number of reads $t$ and segments $\nu$, the decoding complexity scales exponentially with $\tilde{k}$ through the $4^{\tilde{k}}$ term. Nevertheless, $\tilde{k}$ can be chosen sufficiently small, e.g., $\tilde{k}<10$, to keep the computation practical, at the expense of increasing the number of segments~$\nu$ and thus the redundancy introduced by the markers.
\section{Simulation Results} \label{simul}

We evaluate the performance of our MDC-based coding scheme at two levels: the inner code alone and the overall concatenated code. At the level of the inner code, we assess performance through the quaternary symbol error rate (SER), defined as the normalized Hamming distance between the channel input and the output obtained by applying hard-decision decoding to the APPs produced by the inner decoder. At the concatenated-code level, we evaluate performance through the frame error rate~(FER) obtained after applying the BP decoder of the outer LDPC code while using the APPs generated by the inner MDC code. We report $R$ as the symbol rate in sym/nt. Since each quaternary symbol carries two bits, the corresponding information rate in bits/nt is twice the reported value.


To model the error characteristics observed in DNA storage, we fix the total error rate to $p_e=10\%$ and set the edit proportions to \mbox{$(p_s,p_d,p_i) = (0.527,0.447,0.026)p_e$}. These proportions match empirical measurements reported in real experiments~\cite{digitaltwin_errorrates}, and a total edit error rate of $10\%$ corresponds to a high-error regime in practice, consistent with the use of low-fidelity DNA synthesis or sequencing technologies such as photolithographic synthesis~\cite{antkowiak2020_setup} or nanopore sequencing~\cite{zhao2024composite}.

We first evaluate the performance of the inner MDC code in comparison with the convolutional code (CC) used in~\cite{conv_concat} and with marker-repeat (MR) codes as described in~\cite{TBMA}. A common codeword length of $n=160$ is used for all codes, and the performance is evaluated at two code rates, \mbox{$R=0.7$} and $R=0.8$. The CCs are obtained by puncturing the rate-$0.5$ code presented in~\cite{conv_concat}. The MR codes are obtained by repeating symbols once with uniformly spaced periods to match the target rates. For both CC and MR codes, the APP computation is performed using the TBMA algorithm~\cite{TBMA}, with the maximum drift set to $25$ and lookahead enabled.

For $R=0.7$, an $[8,7,2,16]$-MDC code (see Definition~\ref{def}) is constructed using the marker sequence $\x{m}=(0,2)$ and the dual of the length-$8$ repetition code as $\mathcal{C}$. For $R=0.8$, an $[8,8,2,16]$-MDC code is constructed using the same marker and by leaving the message substrings uncoded. The MDC decoding parameters are set as follows for both cases: offset limit $\Phi=25$, per-read shortlist size $|\mathcal{L}_{\tau}|=500$, $\kappa$-mer length $\kappa = 3$, and confidence threshold $F_T=0.985$.

The resulting SERs and average decoding times for CC, MR, and MDC are shown in Fig.~\ref{fig:compare_inner} as functions of the number of reads. The results in Figs.~\ref{fig:1a} and~\ref{fig:1b} show that, for both considered code rates, MDC achieves a lower SER than both CC and MR for $t \geq 3$. Notably, the improvement in SER becomes more pronounced as the number of reads increases, indicating that MDC is able to exploit additional reads more effectively. In addition to the improvement in reliability, the MDC decoder also exhibits faster decoding. In particular, for $R = 0.8$, Fig.~\ref{fig:1c} shows that MDC consistently achieves faster decoding than CC and MR. The decoding time gains are particularly significant when the structure of the MDC code is leveraged and the decoding is parallelized, attaining average decoding times as low as $34$~ms per codeword of length $n=160$ for $t=8$ reads. Overall, these results demonstrate that MDC achieves simultaneous gains in both reliability and decoding time over commonly used inner encoding and decoding schemes for multi-read edit channels.

In the concatenated setting, we consider CC and MDC as inner codes, each paired with an outer LDPC code decoded via BP with 100 iterations. The CC-LDPC combination follows the construction proposed in~\cite{conv_concat}, but uses the TBMA algorithm~\cite{TBMA} for inner-code decoding, which yields more accurate APPs. The resulting FER curves as functions of the number of reads are shown in Fig.~\ref{fig:concat_FER}, and the corresponding concatenated code parameters are summarized in Table~\ref{tab:parameters}. At an overall rate of~$0.5$, the concatenated scheme with MDC as the inner code achieves FER performance comparable to that of the construction with CC. However, at higher rates, namely $0.6$ and~$0.7$, which are more representative of cost-efficient operating regimes for practical DNA storage systems, the MDC-based constructions provide significant gains in FER. These results therefore indicate that the benefits observed at the inner-code level are effectively mirrored in the concatenated setting, particularly at higher rates.

In conclusion, the results show that MDC-based constructions are well suited for multi-read edit channels in the short-blocklength, high-rate regime of interest for DNA storage. The proposed schemes achieve strong reliability at medium-to-high overall rates while maintaining practical decoding complexity, making them effective for efficient use of DNA synthesis and sequencing resources. 

\begin{figure}[!ht]
    \centering
        \begin{tikzpicture}[scale=0.95]
        \begin{axis}[
            width = 0.87\columnwidth,
            height = 0.69\columnwidth,
            xlabel = {Number of Reads $t$},
            ylabel = {Frame Error Rate (FER)},
            scale only axis,
        enlarge x limits=false,
        enlarge y limits=false,
            ymode = log,
            ytick = {1e-5,1e-4,1e-3,1e-2,1e-1,1},
            xtick = {1,2,3,4,5,6},
            xmin = 1, xmax = 6,
            ymin = 1e-5, ymax = 1,
            grid=both,
            major grid style={gray!50},
            minor grid style={gray!30},
            ylabel style = {yshift = -3pt},
            legend columns = 1,
            legend style = { nodes={inner sep=2pt}, 
                             font = \fontsize{6}{7}\selectfont,anchor = north east,
                             at = {(0.99,0.99)},draw opacity = 0.8, fill = white, fill opacity = 0.8,
                             text opacity = 1, cells = {anchor=west}},
            label style = { font = \fontsize{9}{10}\selectfont},
            tick label style={font=\fontsize{9}{10}\selectfont}
        ]
    
        \addplot[
        red,
        thick,
        mark=o,
        mark options={solid, fill=none},
        ]
        coordinates {
            (1,1)(2,0.856)(3,0.299)(4,0.06)(5,0.011)(6,0.0041)
        };
        \addlegendentry{CC-$1$}

        \addplot[
        blue,
        thick,
        mark=o,
        mark options={solid, fill=none},
        ]
        coordinates {
            (1,0.9998)(2,0.7616)(3,0.1034)(4,0.0109)(5,0.0019)(6,0.0005)
        };
        \addlegendentry{MDC-$1$}

        \addplot[
        red,
        thick,
        mark=square,
        mark options={solid, fill=none},
        ]
        coordinates {
            (1,0.99)(2,0.388)(3,0.06)(4,0.0057)(5,0.0013)
        };
        \addlegendentry{CC-$2$}

        \addplot[
        blue,
        thick,
        mark=square,
        mark options={solid, fill=none},
        ]
        coordinates {
            (1,0.985)(2,0.195)(3,0.009)(4,0.0006)(5,0.00011999999999999999)
        };
        \addlegendentry{MDC-$2$}

        \addplot[
        red,
        thick,
        mark=diamond,
        mark options={solid, fill=none},
        ]
        coordinates {
            (1,0.72)(2,0.036)(3,0.0005)(4,3e-5)
        };
        \addlegendentry{CC-$3$}

        \addplot[
        blue,
        thick,
        mark=diamond,
        mark options={solid, fill=none},
        ]
        coordinates {
            (1,0.808)(2,2.8e-02)(3,3e-4)(4,2e-05)
        };
        \addlegendentry{MDC-$3$}
        \end{axis}
    \end{tikzpicture}
    \caption{Frame error rate (FER) as a function of the number of reads $t$, for different concatenated constructions. The channel edit error rate is fixed to $p_e = 10\%$. The parameters of the concatenated schemes are listed in Table~\ref{tab:parameters}. All results are averaged over $10^5$ independent trials. }
    \label{fig:concat_FER}
\end{figure}
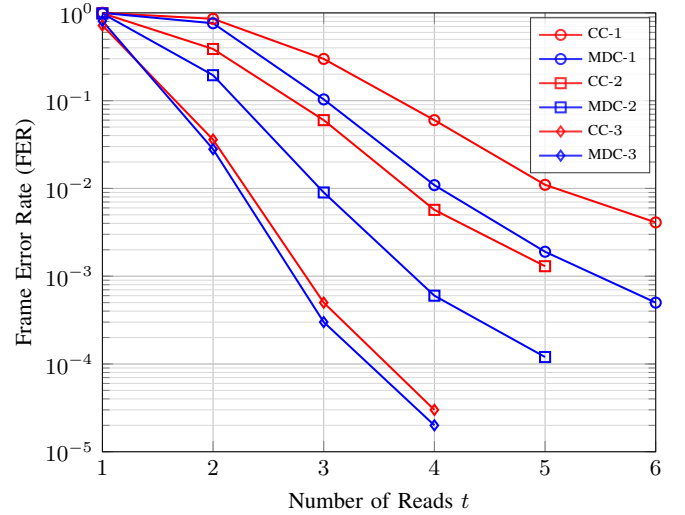

\begin{table}[!ht]
\caption{List of parameters for the concatenated coding schemes.}\label{tab:parameters}
\centering
{\setlength{\tabcolsep}{4pt}
\begin{tabular}{ccccc}
    \toprule
    Scheme & Outer-Code-$[n_o,k_o]$ & Inner-Code & Total $[n,k]$ & Rate $R$ \\
    \midrule
    CC-$1$ & LDPC-$[128,112]$ & CC, $R_{\text{in}} = 0.8$ & $[160,112]$ & $0.7$  \\
    MDC-$1$ & LDPC-$[128,112]$ & MDC-$[8,8,2,16]$ & $[160,112]$ & $0.7$\\
    CC-$2$ & LDPC-$[128,96]$ & CC, $R_{\text{in}} = 0.8$ & $[160,96]$ & $0.6$\\
    MDC-$2$ & LDPC-$[128,96]$ & MDC-$[8,8,2,16]$ & $[160,96]$ & $0.6$ \\
    CC-$3$ & LDPC-$[112,80]$ & CC, $R_{\text{in}} = 0.7$ & $[160,80]$ & $0.5$\\
    MDC-$3$ & LDPC-$[128,80]$ & MDC-$[8,8,2,16]$ & $[160,80]$ & $0.5$ \\
    \bottomrule 
\end{tabular}}
\end{table}

\bibliographystyle{IEEEtran}
\bibliography{references}

@inproceedings{mgcplus,
  title={Marker Guess \& Check Plus ({MGC+}): An Efficient Short Blocklength Code for Random Edit Errors},
  author={Khabbaz, Ramy and Antonini, Marc and {Kas Hanna}, Serge},
  booktitle={2025 13th International Symposium on Topics in Coding (ISTC)},
  pages={1--5},
  year={2025},
  organization={IEEE}
}

@inproceedings{TBMA,
  title={Trellis {BMA}: Coded trace reconstruction on {IDS} channels for {DNA} storage},
  author={Srinivasavaradhan, Sundara Rajan and Gopi, Sivakanth and Pfister, Henry D and Yekhanin, Sergey},
  booktitle={2021 IEEE International Symposium on Information Theory (ISIT)},
  pages={2453--2458},
  year={2021},
  organization={IEEE}
}

@article{conv_concat,
  title={Concatenated codes for multiple reads of a {DNA} sequence},
  author={Maarouf, Issam and Lenz, Andreas and Welter, Lorenz and Wachter-Zeh, Antonia and Rosnes, Eirik and i Amat, Alexandre Graell},
  journal={IEEE Transactions on Information Theory},
  volume={69},
  number={2},
  pages={910--927},
  year={2022},
  publisher={IEEE}
}

@inproceedings{symbolwiseMAP,
  title={Symbolwise {MAP} estimation for multiple-trace insertion/deletion/substitution channels},
  author={Sakogawa, Ryo and Kaneko, Haruhiko},
  booktitle={2020 IEEE International Symposium on Information Theory (ISIT)},
  pages={781--785},
  year={2020},
  organization={IEEE}
}

@inproceedings{peg_theo,
  title={Progressive edge-growth Tanner graphs},
  author={Hu, Xiao-Yu and Eleftheriou, Evangelos and Arnold, D-M},
  booktitle={GLOBECOM'01. IEEE Global Telecommunications Conference (Cat. No. 01CH37270)},
  volume={2},
  pages={995--1001},
  year={2001},
  organization={IEEE}
}

@misc{peg_implement,
  author       = {Roshan Prabhakar and
                  Shubham Chandak and
                  Kedar Tatwawadi},
  title        = {Implementation of Protograph {LDPC} error correction
                   codes},
  month        = sep,
  year         = 2020,
  publisher    = {Zenodo},
  version      = {v1.0},
  doi          = {10.5281/zenodo.4016076},
  url          = {https://doi.org/10.5281/zenodo.4016076}
}

@article{davey_IDS,
  title={Reliable communication over channels with insertions, deletions, and substitutions},
  author={Davey, Matthew C and MacKay, David JC},
  journal={IEEE Transactions on Information Theory},
  volume={47},
  number={2},
  pages={687--698},
  year={2002},
  publisher={IEEE}
}

@article{digitaltwin_errorrates,
  title={A digital twin for {DNA} data storage based on comprehensive quantification of errors and biases},
  author={Gimpel, Andreas L and Stark, Wendelin J and Heckel, Reinhard and Grass, Robert N},
  journal={Nature Communications},
  volume={14},
  number={1},
  pages={6026},
  year={2023},
  publisher={Nature Publishing Group UK London}
}

@article{markertrellis_decoding,
  title={Decoding for channels with insertions, deletions, and substitutions with applications to speech recognition},
  author={Bahl, L and Jelinek, Frederick},
  journal={IEEE Transactions on Information Theory},
  volume={21},
  number={4},
  pages={404--411},
  year={2003},
  publisher={IEEE}
}

@inproceedings{ratzer2005marker,
  title={Marker codes for channels with insertions and deletions},
  author={Ratzer, Edward A},
  booktitle={Annales des t{\'e}l{\'e}communications},
  volume={60},
  number={1},
  pages={29--44},
  year={2005},
  organization={Springer}
}

@article{mansour2010convolutional,
  title={Convolutional decoding in the presence of synchronization errors},
  author={Mansour, Mohamed F and Tewfik, Ahmed H},
  journal={IEEE Journal on Selected Areas in Communications},
  volume={28},
  number={2},
  pages={218--227},
  year={2010},
  publisher={IEEE}
}

@article{church2012_setup,
  title={Next-generation digital information storage in {DNA}},
  author={Church, George M and Gao, Yuan and Kosuri, Sriram},
  journal={Science},
  volume={337},
  number={6102},
  pages={1628--1628},
  year={2012},
  publisher={American Association for the Advancement of Science}
}

@article{goldman2013towards,
    author = {Goldman, Nick and Bertone, Paul and Chen, Siyuan and Dessimoz, Christophe and LeProust, Emily M and Sipos, Botond and Birney, Ewan},
    journal = {Nature},
    number = {7435},
    pages = {77--80},
    publisher = {Nature Publishing Group UK London},
    title = {Towards practical, high-capacity, low-maintenance information storage in synthesized {DNA}},
    volume = {494},
    year = {2013}
}

@article{grass2015_setup,
  title={Robust chemical preservation of digital information on {DNA} in silica with error-correcting codes},
  author={Grass, Robert N and Heckel, Reinhard and Puddu, Michela and Paunescu, Daniela and Stark, Wendelin J},
  journal={Angewandte Chemie International Edition},
  volume={54},
  number={8},
  pages={2552--2555},
  year={2015},
  publisher={Wiley Online Library}
}

@article{antkowiak2020_setup,
  title={Low cost {DNA} data storage using photolithographic synthesis and advanced information reconstruction and error correction},
  author={Antkowiak, Philipp L and Lietard, Jory and Darestani, Mohammad Zalbagi and Somoza, Mark M and Stark, Wendelin J and Heckel, Reinhard and Grass, Robert N},
  journal={Nature Communications},
  volume={11},
  number={1},
  pages={5345},
  year={2020},
  publisher={Nature Publishing Group UK London}
}

@article{sabary2024survey_coding,
  title={Survey for a decade of coding for {DNA} storage},
  author={Sabary, Omer and Kiah, Han Mao and Siegel, Paul H and Yaakobi, Eitan},
  journal={IEEE Transactions on Molecular, Biological, and Multi-Scale Communications},
  volume={10},
  number={2},
  pages={253--271},
  year={2024},
  publisher={IEEE}
}

@article{press2020hedges,
  title={{HEDGES} error-correcting code for {DNA} storage corrects indels and allows sequence constraints},
  author={Press, William H and Hawkins, John A and Jones Jr, Stephen K and Schaub, Jeffrey M and Finkelstein, Ilya J},
  journal={Proceedings of the National Academy of Sciences},
  volume={117},
  number={31},
  pages={18489--18496},
  year={2020},
  publisher={National Academy of Sciences}
}

@article{DNA-MGCP,
    author = {Khabbaz, Ramy and Mateos, J{\'e}r{\'e}my and Antonini, Marc and {Kas Hanna}, Serge},
    doi = {10.64898/2026.03.11.711016},
    journal = {bioRxiv preprint},
    publisher = {Cold Spring Harbor Laboratory},
    title = {{DNA-MGC+}: A versatile codec for reliable and resource-efficient data storage on synthetic {DNA}},
    year = {2026}
}

@article{bar2025scalable,
  title={Scalable and robust {DNA}-based storage via coding theory and deep learning},
  author={Bar-Lev, Daniella and Orr, Itai and Sabary, Omer and Etzion, Tuvi and Yaakobi, Eitan},
  journal={Nature Machine Intelligence},
  volume={7},
  number={4},
  pages={639--649},
  year={2025},
  publisher={Nature Publishing Group UK London}
}

@article{welzel2023dna,
  title={{DNA}-Aeon provides flexible arithmetic coding for constraint adherence and error correction in {DNA} storage},
  author={Welzel, Marius and Schwarz, Peter Michael and L{\"o}chel, Hannah F and Kabdullayeva, Tolganay and Clemens, Sandra and Becker, Anke and Freisleben, Bernd and Heider, Dominik},
  journal={Nature Communications},
  volume={14},
  number={1},
  pages={628},
  year={2023},
  publisher={Nature Publishing Group UK London}
}

@article{cheraghchi2020_codedtrec,
  title={Coded trace reconstruction},
  author={Cheraghchi, Mahdi and Gabrys, Ryan and Milenkovic, Olgica and Ribeiro, Joao},
  journal={IEEE Transactions on Information Theory},
  volume={66},
  number={10},
  pages={6084--6103},
  year={2020},
  publisher={IEEE}
}

@article{banerjee2024sequential,
  title={Sequential decoding of multiple sequences for synchronization errors},
  author={Banerjee, Anisha and Lenz, Andreas and Wachter-Zeh, Antonia},
  journal={IEEE Transactions on Communications},
  volume={72},
  number={11},
  pages={6660--6676},
  year={2024},
  publisher={IEEE}
}

@inproceedings{banerjee2025sequential,
  title={Sequential Decoding of Multiple Traces Over the Syndrome Trellis for Synchronization Errors},
  author={Banerjee, Anisha and Welter, Lorenz and Amat, Alexandre Graell I and Wachter-Zeh, Antonia and Rosnes, Eirik},
  booktitle={ICASSP 2025-2025 IEEE International Conference on Acoustics, Speech and Signal Processing (ICASSP)},
  pages={1--5},
  year={2025},
  organization={IEEE}
}

@inproceedings{batu2004reconstructing,
  title={Reconstructing strings from random traces},
  author={Batu, Tugkan and Kannan, Sampath and Khanna, Sanjeev and McGregor, Andrew},
  booktitle={SODA},
  volume={4},
  pages={910--918},
  year={2004}
}

@article{chase2021approximate,
  title={Approximate trace reconstruction of random strings from a constant number of traces},
  author={Chase, Zachary and Peres, Yuval},
  journal={arXiv preprint arXiv:2107.06454},
  year={2021}
}

@inproceedings{davies2021approximate,
  title={Approximate trace reconstruction: Algorithms},
  author={Davies, Sami and R{\'a}cz, Mikl{\'o}s Z and Schiffer, Benjamin G and Rashtchian, Cyrus},
  booktitle={2021 IEEE International Symposium on Information Theory (ISIT)},
  pages={2525--2530},
  year={2021},
  organization={IEEE}
}

@article{heckel2019characterization,
  title={A characterization of the {DNA} data storage channel},
  author={Heckel, Reinhard and Mikutis, Gediminas and Grass, Robert N},
  journal={Scientific reports},
  volume={9},
  number={1},
  pages={9663},
  year={2019},
  publisher={Nature Publishing Group UK London}
}

@article{zhao2024composite,
  title={Composite hedges Nanopores codec system for rapid and portable {DNA} data readout with high INDEL-Correction},
  author={Zhao, Xuyang and Li, Junyao and Fan, Qingyuan and Dai, Jing and Long, Yanping and Liu, Ronghui and Zhai, Jixian and Pan, Qing and Li, Yi},
  journal={Nature Communications},
  volume={15},
  number={1},
  pages={9395},
  year={2024},
  publisher={Nature Publishing Group UK London}
}

@article{milenkovic2024dna,
  title={{DNA}-based data storage systems: A review of implementations and code constructions},
  author={Milenkovic, Olgica and Pan, Chao},
  journal={IEEE Transactions on Communications},
  volume={72},
  number={7},
  pages={3803--3828},
  year={2024},
  publisher={IEEE}
}

@INPROCEEDINGS{hanna2024GC,
  author={{Kas Hanna}, Serge},
  booktitle={2024 IEEE International Symposium on Information Theory (ISIT)}, 
  title={Short Systematic Codes for Correcting Random Edit Errors in {DNA} Storage}, 
  year={2024},
  volume={},
  number={},
  pages={663-668},
  doi={10.1109/ISIT57864.2024.10619614},
  note={Extended version available at: \url{https://arxiv.org/abs/2402.01244}}
}

@inproceedings{chandak2019improved,
    author = {Chandak, Shubham and Tatwawadi, Kedar and Lau, Billy and Mardia, Jay and Kubit, Matthew and Neu, Joachim and Griffin, Peter and Wootters, Mary and Weissman, Tsachy and Ji, Hanlee},
    booktitle = {2019 57th Annual Allerton Conference on Communication, Control, and Computing (Allerton)},
    organization = {IEEE},
    pages = {147--156},
    title = {Improved read/write cost tradeoff in {DNA}-based data storage using {LDPC} codes},
    year = {2019}
}

@article{yazdi2017portable,
  title={Portable and error-free {DNA}-based data storage},
  author={Yazdi, SM Hossein Tabatabaei and Gabrys, Ryan and Milenkovic, Olgica},
  journal={Scientific Reports},
  volume={7},
  number={1},
  pages={5011},
  year={2017},
  publisher={Nature Publishing Group UK London}
}

@article{cai2021correcting,
  title={Correcting a single indel/edit for {DNA}-based data storage: Linear-time encoders and order-optimality},
  author={Cai, Kui and Chee, Yeow Meng and Gabrys, Ryan and Kiah, Han Mao and Nguyen, Tuan Thanh},
  journal={IEEE Transactions on Information Theory},
  volume={67},
  number={6},
  pages={3438--3451},
  year={2021},
  publisher={IEEE}
}

@article{tang2024correcting,
  title={Correcting a substring edit error of bounded length},
  author={Tang, Yuanyuan and Motamen, Sarvin and Lou, Hao and Whritenour, Kallie and Wang, Shuche and Gabrys, Ryan and Farnoud, Farzad},
  journal={IEEE Transactions on Communications},
  volume={73},
  number={1},
  pages={12--21},
  year={2024},
  publisher={IEEE}
}

@article{Declercq2007NBLDPC,
  author    = {David Declercq and Marc P. C. Fossorier},
  title     = {Decoding Algorithms for Nonbinary {LDPC} Codes over GF(q)},
  journal   = {IEEE Transactions on Communications},
  volume    = {55},
  number    = {4},
  pages     = {633--643},
  year      = {2007},
  month     = apr,
  doi       = {10.1109/TCOMM.2007.894088},
  issn      = {0090-6778}
}

@article{ukkonen1992approximate,
  title={Approximate string-matching with q-grams and maximal matches},
  author={Ukkonen, Esko},
  journal={Theoretical Computer Science},
  volume={92},
  number={1},
  pages={191--211},
  year={1992},
  publisher={Elsevier}
}

@INPROCEEDINGS{10161631,
  author={Welter, Lorenz and Maarouf, Issam and Lenz, Andreas and Wachter-Zeh, Antonia and Rosnes, Eirik and Amat, Alexandre Graell I},
  booktitle={2023 IEEE Information Theory Workshop (ITW)}, 
  title={Index-Based Concatenated Codes for the Multi-Draw {DNA} Storage Channel}, 
  year={2023},
  volume={},
  number={},
  pages={383-388},
  doi={10.1109/ITW55543.2023.10161631}}

\end{document}